\documentclass[journal,onecolumn]{IEEEtran}

\usepackage[utf8]{inputenc}
\usepackage{verbatim}
\usepackage{optidef}

\usepackage{amsmath,amssymb,amsthm,amsfonts,mathrsfs}
\usepackage{tablefootnote}
\usepackage[margin=0.5in]{geometry}

\usepackage{hyperref} 

\usepackage{graphicx,color,cases} 
\usepackage{xfrac}  
\usepackage{xargs}

\usepackage{tikz}	
\usetikzlibrary{backgrounds,fit,decorations.pathreplacing}  
\usetikzlibrary{automata} 

\usepackage[colorinlistoftodos,prependcaption]{todonotes}

\definecolor{DarkGreen}{rgb}{0.1,0.5,0.1}
\definecolor{DarkRed}{rgb}{0.5,0.1,0.1}
\definecolor{DarkBlue}{rgb}{0.1,0.1,0.5} 
  
\def\>{\rangle} 
\def\<{\langle}

\newtheorem{definitionenv}{Definition}
\newtheorem{lemmaenv}[definitionenv]{Lemma}
\newtheorem{theoremenv}[definitionenv]{Theorem}
\newtheorem{corollaryenv}[definitionenv]{Corollary}
\newtheorem{propositionenv}[definitionenv]{Proposition}
\newtheorem{conjectureenv}[definitionenv]{Conjecture}
\newtheorem{exampleenv}{Example}
\newtheorem{app-lemmaenv}[section]{Lemma}

\newenvironment{definition}{\begin{definitionenv}\rm}{\end{definitionenv}}
\newenvironment{lemma}{\begin{lemmaenv}\rm}{\end{lemmaenv}}
\newenvironment{theorem}{\begin{theoremenv}\rm}{\end{theoremenv}}
\newenvironment{corollary}{\begin{corollaryenv}\rm}{\end{corollaryenv}}
\newenvironment{example}{\begin{exampleenv}\rm}{\end{exampleenv}}
\newenvironment{proposition}{\begin{propositionenv}\rm}{\end{propositionenv}}
\newenvironment{conjecture}{\begin{conjectureenv}\rm}{\end{conjectureenv}}
\newenvironment{app-lemma}{\begin{app-lemmaenv}\rm}{\end{app-lemmaenv}}

\newcommand{\bd}{\begin{definition}}
\newcommand{\ed}{\end{definition}}
\newcommand{\bl}{\begin{lemma}}
\newcommand{\el}{\end{lemma}}
\newcommand{\elp}{\hspace*{\fill} $\Box$
                 \end{lemma}}
\newcommand{\bt}{\begin{theorem}}
\newcommand{\et}{\end{theorem}}
\newcommand{\etp}{\hspace*{\fill} $\Box$
                 \end{theorem}}
\newcommand{\bc}{\begin{corollary}}
\newcommand{\ec}{\end{corollary}}
\newcommand{\ecp}{\hspace*{\fill} $\Box$
                 \end{corollary}}
\newcommand{\bcj}{\begin{conjecture}}
\newcommand{\ecj}{\end{conjecture}}

\newcommand{\be}{\begin{example}}
\newcommand{\ee}{\end{example}}
\newcommand{\eep}{\hspace*{\fill} $\Box$
                 \end{example}}
\newcommand{\bp}{\begin{proposition}}
\newcommand{\ep}{\end{proposition}}
\newcommand{\epp}{
                 \end{proposition}}

\newcommand{\bx}{{\bf x}}

\newcommand{\cG}{{\cal G}}

\newcommand{\cK}{{\cal K}}

\newcommand{\cM}{{\cal M}}
\newcommand{\cN}{{\cal N}}

\newcommand{\cQ}{{\cal Q}}

\newcommand{\mC}{{\mathbb C}}

\newcommand{\ket}[1]{|#1\rangle}

\newcommand{\tr}{\text{tr}}
\newcommand{\Tr}[1]{\text{Tr}\left(#1\right)}
\newcommand{\wt}[1]{\text{wt}\left(#1\right)}

\newcommand{\ot}{\otimes}

\def\picode{\texttt{picode}}

\newcommand{\psiew}{\ket{{\texttt{AUX}}}}

\newcommand{\supewe}{}

\newcommand{\ssl}{^{\texttt{SL}}}

\newcommand{\zo}{\{0,1\}}

\newcommandx{\yellownote}[2][1=]{\todo[linecolor=yellow,backgroundcolor=yellow!25,bordercolor=yellow,#1]{#2}}
\newcommandx{\greennote}[2][1=]{\todo[inline,linecolor=olive,backgroundcolor=green!25,bordercolor=olive,#1]{#2}}

\begin{document}
\title{Linear programming bounds for quantum channels acting on quantum error-correcting codes}
\author{%
Yingkai Ouyang and  Ching-Yi Lai
\thanks{\footnotesize 
This article was presented in part at ISIT 2020.

YO is with the Department of Electrical and Computer Engineering, National University of Singapore, Singapore. Part of work was completed when YO was at the Department of Physics \& Astronomy,
  University of Sheffield,
  Sheffield, S3 7RH, United Kingdom. email: oyingkai@gmail.com

CYL is with the Institute of Communications Engineering, National Yang Ming Chiao Tung University, Hsinchu 30010, Taiwan. email:cylai@nycu.edu.tw
	 }}

\maketitle

\begin{abstract}
While quantum weight enumerators establish some of the best upper bounds on the minimum distance of quantum error-correcting codes, these bounds are not optimized to quantify the performance of quantum codes under the effect of arbitrary quantum channels that describe bespoke noise models.
Herein, for any Kraus decomposition of any given quantum channel, we introduce corresponding quantum weight enumerators that naturally generalize the Shor-Laflamme  quantum weight enumerators.
We establish an indirect linear relationship between these generalized quantum weight enumerators by introducing an auxiliary exact weight enumerator that completely quantifies the quantum code's projector, and is independent of the underlying noise process. 
By additionally working within the framework of approximate quantum error correction,
we establish a general framework for constructing a linear program that is infeasible whenever approximate quantum error correcting codes with corresponding parameters do not exist. 
Our linear programming framework allows us to establish the non-existence of certain quantum codes that approximately correct amplitude damping errors, 
and obtain non-trivial upper bounds on the maximum dimension of a broad family of permutation-invariant quantum codes.
\end{abstract}


\section{Introduction}
\label{sec:introduction}

The distance of an error-correcting code is of central importance in coding theory, because it quantifies the number of adversarial errors that can be corrected. 
For codes of fixed length and rate, upper and lower bounds on their distance can be determined.
The best lower bounds can be obtained from various randomized code constructions that yield the Gilbert-Varshamov bound~\cite{MS77} and this is also true in the quantum case~\cite{FeM04,JiX11,ouyang2014concatenated}).
On the contrary, markedly different techniques are used to derive upper bounds. 
In classical coding theory, weight enumerators count the weight distribution of codewords in a code~\cite{MS77}. 
The MacWilliams identity establishes a linear relationship between the weight enumerators of a code and that of its dual code. 
This allows one to obtain upper bounds on the distance of codes by linear programming. (This may be improved via the Terwilliger algebra and semidefinite programming \cite{schrijver2005new}.)
Further extensions of this technique leads to the celebrated algebraic linear programming bounds \cite{delsarte1973algebraic,Aal90}.

The notion of weight enumerators in the quantum setting is less obvious, because quantum codes on $n$ qubits are subspaces of $\mathbb C^{2^n}$, and these subspaces do not in general admit a combinatorial interpretation.
Shor and Laflamme nonetheless introduced a meaningful definition of weight enumerators for quantum codes \cite{SL97} in terms of the codes' projectors $P$ and a nice error basis for matrices.
In particular, the Shor-Laflamme (SL) quantum weight enumerators are sums of terms of the form $|\tr(EP)|^2$ and $\tr(E P E^\dagger P)$, respectively, where the sums are performed over all Paulis $E$ of a given weight.
We will call the vectors of these enumerators labeled by Pauli weights the A-type and B-type quantum weight enumerators, respectively.
Shor and Laflamme showed that the A-type and B-type quantum weight enumerators are still linearly related in a way reminiscent of the classical relationship~\cite{SL97}. The relation between the two enumerators is the quantum analogue of the famous MacWilliams identity. 
Variations on the SL enumerators were then studied by Rains, which allowed better bounds on the parameters of quantum codes~\cite{Rains98a}.
Because of the existence of a linear relationship between the two types of enumerators, 
linear programming techniques can be applied to establish upper bounds on the minimum distance for (small) quantum stabilizer codes~\cite{CRSS98}. Algebraic linear programming bounds based on the MacWilliams identity, such as the Singleton, Hamming, and the first linear programming bounds, are also derived for general quantum codes~\cite{AL99}. These results have been extended to entanglement-assisted quantum stabilizer codes~\cite{LBW13,LA18} and quantum data-syndrome codes~\cite{ALB20}.
Also there is a MacWilliams identity for (entanglement-assisted) quantum convolutional codes~\cite{LHL16}.
Recently it was shown that the SL weight enumerators of a codeword stabilized quantum code has an interpretation as the enumerator of an associated classical code~\cite{NK21}.

Although the distance of a quantum code is a meaningful metric with respect to adversarial noise, estimates on the performance of a quantum code derived from the distance under specific noise models are often overly pessimistic.
For instance, while a minimum of five qubits is needed to perfectly correct an arbitrary error~\cite{SL97}, four qubits suffice to correct a single \textit{amplitude damping}  (AD) error~\cite{LNCY97}.
However, most quantum weight enumerators give no direct result regarding limits on the ultimate performance of quantum codes under the influence of general quantum channels, even in the simple case of AD errors. 
To better understand these fundamental limits, it would be advantageous to have MacWilliams-type identities for different quantum weight enumerators defined for various noisy quantum channels, from which corresponding linear programming bounds can be obtained.
Currently, most linear programming bounds for quantum codes use quantum weight enumerators only describe quantum error correction in the perfect setting~\cite{Rains98a}. Because of this, these methods do not readily extend to quantum codes under the action of arbitrary quantum channels and in the paradigm of approximate quantum error correction (AQEC). 

To address the aforementioned problems, we extend the theory of quantum weight enumerators to deal with  AQEC codes for any given quantum channel. 
Namely, we generalize the two SL quantum weight enumerators to address quantum codes under the influence of any set of Kraus operators. 
This goes beyond the theory  that the authors previously introduced in Ref \cite{ouyang2020linear}, where only amplitude damping errors were discussed.
While we do not have a MacWilliams identity that establishes a direct linear relationship between these two generalized quantum weight enumerators, we do establish an indirect linear relationship between them.
To enable this, we rely on an \textit{auxiliary exact weight enumerator} with respect to Pauli operators,
which is exact in the sense that it depends explicitly on the matrix decomposition of the code projector $P$ in the Pauli basis.
We thereby show linear connections between this enumerator and our two generalized quantum weight enumerators. 
This allows us to establish a linear program that is infeasible only when AQEC codes do not exist. 

To illustrate the utility of our framework, we apply it in two different scenarios. 
In the first scenario, we establish the non-existence of quantum codes that approximately correct amplitude damping errors. 
In particular, we numerically rule out the existence of three-qubit AQEC AD codes that are capable  of  correcting  an  arbitrary  AD  error.
Our linear program cannot eliminate the existence of a four-qubit code that can correct one AD error and this agrees to the four-qubit AD code proposed in~\cite{LNCY97}. 
In the second scenario, we show how our framework can be adapted to quantum codes that must be permutation-invariant, and provide linear-programming bounds for the non-existence of permutation-invariant codes of prescribed distance.

This paper is organized as follows.
In Sec.~\ref{sec:prelims}, we review notation for Pauli operators (Sec.~\ref{subsec:paulis}), quantum channels (Sec.~\ref{subsec:quantum-channels}), quantum codes and weight enumerators (Sec.~\ref{subsec:SL-enumerators}), and the space of complex square matrices (Sec.~\ref{subsec:matrix-space}). 
In Sec.~\ref{sec:Kraus-wtenum}, we introduce our quantum weight enumerators for general Kraus operators in AQEC.
In Sec.~\ref{sec:aux}, we introduce auxiliary weight enumerators, and in Sec.~\ref{sec:connection-matrices}, we propose connection matrices that establish linear relationships between our quantum weight enumerators and the auxiliary weight enumerators.
In Sec.~\ref{sec:linear program bounds}, we formulate a linear program for general quantum channels.
We discuss applications of the linear program bounds for AD errors in Sec.~\ref{subsec:AD}
and for permutation-invariant quantum codes in Sec.~\ref{subsec:picodes}.
We conclude our results 
in Sec.~\ref{sec:discussions}, 

\section{Preliminaries} \label{sec:prelims}
 
\subsection{Pauli Operators} \label{subsec:paulis}
A single-qubit state space is a two-dimensional complex Hilbert space $\mathbb{C}^2$,
and a multiple-qubit state space is simply the tensor product space of single-qubit spaces $(\mathbb C^2)^{\otimes n}=\mathbb C^{2^n}$.
The Pauli matrices
$$\left\{I_2=\begin{bmatrix}1 &0\\0&1\end{bmatrix}, X=\begin{bmatrix}0 &1\\1&0\end{bmatrix},  Z=\begin{bmatrix}1 &0\\0&-1\end{bmatrix}, Y=iXZ\right\}$$
form a basis of the linear operators on $\mathbb{C}^2$.
Let  
\[{\cG}_n=\{M_1\otimes M_2\otimes \cdots \otimes M_n: M_j\in\{I_2,X,Y,Z\} \},\] which is a basis of the linear operators on the $n$-qubit state space $\mathbb{C}^{2^n}$.
The weight of an element  $E=M_1\otimes  \cdots \otimes M_n$ in $\cG_n$, denoted $\wt{E}$,   is the number of $M_j$'s that are non-identity matrices.

\subsection{Quantum channels} \label{subsec:quantum-channels}
A quantum channel that takes an $n$-qubit state to an $n$-qubit state is a completely positive and trace-preserving linear map from $L(\mathbb  C^{2^n})$ to $L(\mathbb  C^{2^n})$, where $L(\mathbb C^{2^n})$ is the set of all linear maps from 
$\mathbb C^{2^n}$ to $\mathbb C^{2^n}$.
In particular, every quantum channel $\mathcal N$ admits a (non-unique) decomposition into Kraus operators $\Omega=\{W\}\subset L(\mathbb  C^{2^n})$ such that for any $2^n\times 2^n$ matrix $M$, we have
\begin{align}
\mathcal N(M)  
=
\sum_{W \in \Omega} 
    W  M W^\dagger ,
\end{align}
where 
\begin{align}
    \sum_{W \in \Omega} W^\dagger  W = I_{2^n},
\end{align}
where $I_{2^n}$ denotes the $2^n\times 2^n$ identity matrix.

\subsection{Quantum Codes and Weight Enumerators}
\label{subsec:SL-enumerators}

An $n$-qubit quantum code $\cQ$ is a subspace of $\mC^{2^n}$.
Let $P$ denote the codespace projector onto $\cQ$
The quantum code $\cQ$ satisfies the Knill-Laflamme quantum error correction criterion \cite{KnL97} with \textit{minimum distance} $d$  if and only if
\begin{align}
    P E P = g_E P \label{KL-matrix-condition}
\end{align}
for some complex coefficients $g_E$ for every $E \in \mathcal G_n$ of weight at most $d-1$.
This means that any Paulis of weight at most $d-1$ can be detected since it brings logical codewords to orthogonal subspaces. An $n$-qubit quantum code of dimension $M$ and minimum distance $d$ is denoted $((n,M,d))$. If $M=2^k$,
it is denoted $[[n,k,d]]$.

Shor and Laflamme defined two  weight enumerators $\{A\ssl_{i}\}$ and $\{B\ssl_{i}\}$ of  $\cQ$  by
\begin{align}
A\ssl_{i}=&\frac{1}{\tr(P)^2}\sum_{ E\in {\cG}_n, \wt{E}=i } \Tr{E P}\Tr{E^\dag P}, \label{eq:Bij}
\end{align}
and
\begin{align}
B\ssl_{i}=&\frac{1}{\tr(P)}\sum_{ E\in {\cG}_n, \wt{E}=i} \Tr{E PE^\dag P}, \label{eq:Bij_perp}
\end{align}
for $i=0,\dots,n$ \cite{SL97}.
These two weight enumerators will be called \emph{SL enumerators} in this article.
 
  Since ${\cG}_n$ is a basis for the linear operators on $\mathbb{C}^{2^n}$, we have $P=\sum_{E \in \mathcal G_n}a_E E$, where $a_E \in \mathbb R$ because $P$ is Hermitian.
It follows that $\tr(EP) = a_E.$
Since $P$ is a projector, we have $P=P^2$, and it follows that $\tr(P) = \tr(P^2) = 2^n\sum_{E\in \mathcal G_n} a_E^2$.
Hence
\begin{align}
    \sum_i A_i\ssl
= \frac{1}{\tr(P)^2}\sum_E a_E^2.
= \frac{2^n\tr(P)}{\tr(P)^2} = \frac{2^n}{\tr(P)}.
\end{align}
Also $A_0\ssl = 1$.

The power of the SL enumerators is that the (perfect) quantum error correction criterion of Knill and Laflamme are equivalent to certain linear constraints on these SL enumerators.
This is because $A\ssl_i=B\ssl_i$ if and only if \eqref{KL-matrix-condition} holds for every $E \in \mathcal G_n$ of weight $i \leq d-1$.

\subsection{The space of complex square matrices}
\label{subsec:matrix-space}
The Hilbert-Schmidt inner product of two square complex matrices  $U,V$ is defined by 
\begin{align}
    \<U,V\>=\tr(U^\dagger V),
\end{align}
 and induces a norm, called the the Frobenius norm. Namely, the Frobenius norm of $U$ is defined as $\| U \|_F = \sqrt{ \<U,U\>}$.

Given any complex square matrices $U$ and $V$ of the same size, 
we can use the Gram-Schmidt process to get
\begin{align}
    W = U - R,
\end{align}
where
\begin{align}
    R = \< V, U \> \frac{V}{\| V \|_F}
\end{align}
denotes the component of $U$ which is parallel to $V$, and $W$ denotes the component of $U$ which is orthogonal to $ V/ \| V \|_F$.  
Geometrically, $R$ and $W$ are orthogonal, and they satisfy the Pythagoras theorem in the sense that 
\begin{align}
    \|R\|_F^2 + \|W\|_F^2 = \|U\|_F^2.
\end{align}
Now using the fact that 
    $\|R\|_F^2 = |\<V,U\>|^2$, 
it follows that 
\begin{align}
|\<V,U\>|^2  = \|U\|_F^2 - \|W\|_F^2. 
\label{eq:pythagoras}
\end{align}
 
\section{Quantum weight enumerators for sets of Kraus operators}
\label{sec:Kraus-wtenum}

In what follows, we generalize the SL enumerators to allow direct consideration of an arbitrary set of Kraus operators. 
For the SL enumerators, the error operators considered are the Pauli operators, which form a nice error basis. In generalizing these Pauli operators to general Kraus operators, we will no longer be able to leverage many properties that the nice error basis affords. 
(In particular, it is unknown how to generalize the MacWilliams identity between the SL enumerators.)
We nonetheless can generalize the definition of SL enumerators to general Kraus operators.

Let $P$ be the projector onto a quantum code, and let $\Omega$ be a set of Kraus operators for a quantum channel.
In general, the set of Kraus operators $\Omega$  does not necessarily span the space of linear operators on $\mathbb{C}^{2^n}$, and need not even be a basis for $\mathbb{C}^{2^n}$.
 Suppose that $\Omega$ is partitioned into the disjoint sets $\Omega_0, \dots, \Omega_{w-1}$  in accordance to the severity of the Kraus operators therein.
Here, $w$ counts the number of such sets.
We define two  enumerators (vectors) with coefficients
\begin{align}
    A_i    &= 
    \frac{1}{(\tr P)^2}
    \sum_{E \in \Omega_i} \tr(EP)\tr(E^\dagger P),\ i=0,\dots,w-1,
    \label{Aenum} \\ 
    B_i     &= 
    \frac{1}{\tr P}
    \sum_{E \in \Omega_i} \tr(EPE^\dagger P) , \ i=0,\dots,w-1,\label{Benum}
\end{align}
respectively.
In what follows, we use the Dirac ket notation to represent weight enumerators as in \cite{LHL16}.
We denote the $A$ and $B$-type enumerators as 
\begin{align}
    |A\> &= \sum_{i=0}^{w-1} A_i |i\>,\\
    |B\> &= \sum_{i=0}^{w-1} B_i |i\>.
\end{align}
Note that 
\begin{align*}
    \sum_i B_i =& \frac{1}{\tr P}
    \sum_{E \in \Omega} \tr(EPE^\dagger P) \\
    \leq&  \frac{1}{\tr P} \sum_{E \in \Omega} \tr(EPE^\dagger) =\frac{1}{\tr P}  \tr\left( \left(\sum_{E \in \Omega} E^\dagger E\right) P\right) = 1,
\end{align*}
where the inequality is because both $EPE^\dag$ and $P$ are positive semidefinite and $P\leq I_{2^n}$. 

Hence we have
\begin{align}
    B_0 + \cdots + B_{w-1} \le 1.
\end{align}
Another interpretation of this inequality is that the sum of the coefficients of the $B$-type enumerator retains interpretation as the fidelity of a quantum code after the action of the quantum channel with Kraus operators in $\Omega$ without quantum error correction. 

It can be shown, as in \cite{AL99}, that 
\begin{align}
B_i\geq A_i, \quad i=0,\dots,w-1. \label{eq:BgeqA}
\end{align}
Furthermore, since the code projector $P$ is Hermitian and we have the cyclic property of the trace, it is clear that
$\tr(EP)\tr(E^\dagger P) = \tr(EP)\tr(P^\dagger E^\dagger) = |\tr(EP)|^2$. This implies that $A_i$ is always a sum of non-negative terms, and hence we must have
\begin{align}
    A_i \ge 0 , \quad i = 0,\dots, w-1.
\end{align}

In this paper, we will define approximate quantum error correction using the language of quantum weight enumerators.
In fact, we will see that $B_i-A_i = 0$ is equivalent to saying that the Knill-Laflamme quantum error criterion is satisfied for every Kraus operator in the set $\Omega_i$.
When $B_i-A_i$ is non-zero, we quantify this non-zero quantity precisely in terms of the deviations from the Knill-Laflamme conditions in the Frobenius norm.

 For every $E \in \Omega$, let 
\begin{align}
\mathcal E_E
    = 
PEP - \<PEP,P\>
\frac{P}{\sqrt{\tr(P)}}.
\end{align}
Then we have the following lemma.
\begin{lemma}[Quantum weight enumerators and approximate quantum error correction]
\label{lem:aqec}
Let $A_i$ and $B_i$ be quantum weight enumerators defined by \eqref{Aenum} and \eqref{Benum}, respectively,
using the code projector $P$ and the subsets $\Omega_0,\dots, \Omega_{w-1}$.
Then for every $i = 0,\dots, w-1$, we have 
\begin{align}
B_i-A_i
=
\frac{1}{\tr(P)}
\sum_{E \in \Omega_i}    \|\mathcal E_E \|_F^2.
\end{align}
\end{lemma}
\begin{proof}
By identifying $U = PEP$ and $V=P$ and substituting into \eqref{eq:pythagoras}, we get
\begin{align}
|\<PEP,P\>|^2 =
   \|PEP\|_F^2 \tr(P) 
   - \|\mathcal E_E \|_F^2\tr(P).
  \label{pythagoras_2}
\end{align}
Rewriting \eqref{pythagoras_2}, we get
\begin{align}
|\<PEP,P\>|^2 
=
\tr(EPE^\dagger P) \tr(P) 
- \|\mathcal E_E \|_F^2\tr(P).
\label{pythagoras_3}
\end{align}
Since $|\<PEP,P\>|^2 = |\tr(EP)|^2$, this implies that
\begin{align}
A_i
=
B_i 
- 
\frac{1}{\tr(P)}
\sum_{E \in \Omega_i}    \|\mathcal E_E \|_F^2,
\end{align}
from which the lemma follows.
\end{proof}

From this lemma, we obtain a perturbed quantum error correction criterion using the language of quantum weight enumerators.
In particular, we can say that $B_i-A_i$ is  proportional to the sum of the squares of the Frobenius norms of $\mathcal E_E$, where $E\in \Omega_i$.
The matrices $\mathcal E_E$ have been previously studied in Ref.~\cite{ouyang2014permutation}, where their connection to the infidelity of a quantum code is elucidated.

Given the form of Lemma \ref{lem:aqec}, we define a notion of AQEC with respect to the perturbation of the KL conditions as follows.
\bd
\label{def:epsilon-code}
A quantum code with code projector $P$ is $(\epsilon_0,\dots, \epsilon_{w-1})$-AQEC with respect to the sets $\Omega_0,\dots, \Omega_{w-1}$ if 
\begin{align}
   \frac{1}{\tr(P)}
\sum_{E \in \Omega_i}    \|\mathcal E_E \|_F^2 \le \epsilon_i
\end{align}
for all $i=0,\dots, w-1$.
\ed

From this context, perturbations to the Knill-Laflamme quantum error correction criterion $(B\ssl_i-A\ssl_i>0)$ can be understood by directly perturbing the linear constraints on the quantum weight enumerators.

\section{Auxiliary weight enumerators} \label{sec:aux}
Without the existence of a MacWilliams identity, we can nonetheless establish a linear relationship between $|A \>$ and $|B \>$ by introducing additional vectors that reside on an auxiliary space.
Recall that the projector $P$ of a quantum code, when decomposed in the Pauli basis, can be written as
\begin{align}
P=&  \sum_{\sigma\in {\cG}_n} 
\frac{\tr(\sigma P)}{2^{n}} \sigma. \label{eq:P}
\end{align}
\begin{definition}
The \emph{auxiliary (exact) weight enumerator} corresponding to code projector $P$ is  given by
\begin{align}
    \psiew
    &= |\phi\> \otimes |\phi\>  ,
    \label{def:aux}
\end{align}
where
\begin{align}
|\phi\> = \sum_{\sigma  \in  {\cG}_n} \tr(\sigma P)  |\sigma\> .
\end{align}
\end{definition}
The auxiliary weight enumerator is exact in the sense that it encompasses complete information about the quantum code's projector. 
We emphasize that the state $|\phi\>$ depends only  on the code's projector $P$.
Hence $|\phi\>$ is independent of the  channel in consideration. 

Define a swap operation  \begin{align}
    \Pi_{} 
    &= \sum_{\sigma, \tau \in {\cG}_n}
    |\sigma\>\<\tau| \otimes |\tau\>\<\sigma|.
\end{align}
It follows that  $\psiew$ is an eigenvector of $\Pi$ with eigenvaule $+1$
\begin{align}
\Pi_{} \psiew &= \psiew   ,
\end{align} 
since $\tr(\sigma P)\tr(\tau P)$ is invariant under the swap of $\sigma$ and $\tau$.
We later exploit this permutation symmetry to introduce additional constraints in our linear program for amplitude damping channels.
\section{Connection Matrices}\label{sec:connection-matrices}
To establish the connection between our auxiliary weight enumerator $\psiew$ with the two generalized weight enumerators $\ket{A}$ and $\ket{B}$, we define two matrices 
as follows:
\begin{align}
M_A\supewe &= 
    \sum_{i=0}^{w-1} \sum_{E \in \Omega_i}
    \sum_{\sigma , \tau \in {\cG}_n}
    2^{-2n}
    \tr(E \sigma )
    \tr(E^\dagger \tau)
    |i\> \<\sigma|\<\tau|,\label{eq:MA-defi}\\
M_B \supewe 
    &= 
    \sum_{i=0}^{w-1}   \sum_{E \in \Omega_i}
    \sum_{\sigma , \tau \in {\cG}_n}
    2^{-2n}
    \tr(E \sigma E^\dagger \tau)
    |i\> \<\sigma|\<\tau|
    \label{eq:MB-defi}.
\end{align}
The matrices $M_A$ and $M_B$  establish an indirect linear relationship between the generalized enumerators $|A\>$ and $|B\>$ via an additional linear relationship with the auxiliary weight enumerator. 
Namely, we have the following linear relationships.
\begin{lemma} 
\label{lem:connection-matrices}
The following matrix identities hold.
\begin{align}
M_A \supewe  \psiew &=(\tr P )^2 |A \>, 
\label{eq:MAewe}\\
M_B \supewe \psiew &= \tr P    |B \>. \label{eq:MBewe}  
\end{align}
\end{lemma}
\begin{proof} 
By   (\ref{eq:P}), 
we get
\begin{align}
    |A  \>   =&
    \frac{1}{(\tr P )^2}
    \sum_{i=0}^{w-1}
    \sum_{E \in \Omega_i} 2^{-2n}
    \sum_{\sigma, \tau \in {\cG}_n} 
    \tr(E\sigma)\tr(E^\dagger \tau)
    \notag\\
    &\quad \times
    \tr(\sigma P) \tr( \tau P) |i\>.
\end{align}
Also we can see that 
\begin{align}
    M_A \supewe \psiew
    =
    \sum_{i=0}^{w-1}
    \sum_{\substack{
    E \in \Omega_i\\
    \sigma , \tau \in {\cG}_n
    }}
  \frac{  \tr(E \sigma )
    \tr(E^\dagger \tau) }{2^{2n} }
    |i\>  
    \tr(\sigma P) 
    \tr(\tau P)  . 
\end{align}
Hence \eqref{eq:MAewe} holds. 

To obtain the second identity,
we also expand the code projector $P$ in the Pauli basis to get
\begin{align}
    |B \>   &=
    \frac{1}{\tr(P)}
    \sum_{i=0}^{w-1}
    \sum_{E \in \Omega_i} 2^{-2n}
    \sum_{\sigma, \tau \in {\cG}_n} 
    \tr(E\sigma E^\dagger \tau)
    \tr(\sigma P) \tr( \tau P) |i\>.
\end{align}
Next, note that 
\begin{align}
    M_B \psiew
    =
    \sum_{i=0}^{w-1}     \sum_{\substack{
    E \in \Omega_i\\
    \sigma , \tau \in {\cG}_n
    }}
    2^{-2n}
    \tr(E \sigma E^\dagger \tau)
    |i\>  
    \tr(\sigma P) 
    \tr(\tau P)  .
\end{align}
The result 
$M_B \supewe  \psiew = \tr(P)    |B \>$ then follows.
\end{proof}
While we do not have a direct linear relationship between the generalized quantum weight enumerators $|A\>$ and $|B\>$, Lemma~\ref{lem:connection-matrices} establishes a linear relationship between each generalized quantum weight enumerator and the auxiliary weight enumerator. This thereby establishes an indirect linear relationship between $|A\>$ and $|B\>$, which allows us to establish linear programming bounds for AD codes  later. 

It is also important to note the following properties of connection matrices.
\begin{enumerate}
    \item The connection matrices $M_A\supewe$ and $M_B\supewe$ are devoid of information about the code, because they are both independent of the code projector $P$.
\item The connection matrices $M_A\supewe$ and $M_B\supewe$ depend on the set of Kraus operators $\Omega$ that describe the underlying quantum channel.
\end{enumerate}

\section{Linear Programming Bounds for general quantum channels} \label{sec:linear program bounds}

Here, given a quantum channel with a set of Kraus operators $\Omega$, partitioned into disjoint subsets $\Omega_{0}, \dots, \Omega_{w-1}$, we introduce a linear program with optimization variables  $A_0,\dots,A_{w-1}$, $B_0,\dots,B_{w-1}$,
which are non-negative. 
The constraints in this linear program arise from relating the $A$- and $B$-type quantum weight enumerators introduced in \eqref{Aenum} and \eqref{Benum}. 
While the $A$- and $B$-type enumerators do not necessarily have a direct linear relation to one another, 
they are both directly linearly related to the auxiliary exact weight enumerator of a quantum code \eqref{def:aux} via the connection matrices \eqref{eq:MA-defi} and \eqref{eq:MB-defi}. 
Infeasibility of this linear program allows us to to establish the non-existence of certain AQEC quantum codes. 
Since it is only the feasibility of the linear program that is important, we can always set the objective function of the linear program to be a constant, that is for instance~0.

If a quantum code is $(\epsilon_0, \dots, \epsilon_{w-1})$-AQEC with respect to the sets $\Omega_0,\dots, \Omega_{w-1}$,
then the following linear constraints admits a feasible solution.
\begin{align}
    {\rm Find}\ \ A_0,\dots,&A_{w-1},B_0,\dots, B_{w-1}, \psiew\notag\\
      {\rm subject\ to\ } 
(\tr P)^2|A  \>      &=M_A 
\psiew \notag\\
\tr P|B  \> &=M_B   \psiew\notag\\
0\leq B_i    - A_i  &\le \epsilon_i  , \quad 0 \le i \le w-1\notag\\ 
B_0+\dots +B_n      &\le 1 \notag\\ 
A_i      &\ge 0, \quad 0 \le i \le w-1  \notag\\
\Pi \psiew &= \psiew. \label{eq:linear-constraints}
\end{align}
Here at this abstract level, the distance $d$ does not appear. If we work with the SL enumerators, then we have $\epsilon_0 = \cdots= \epsilon_{d-1}=0$.
We have thereby derived a linear programming bound that applies to any quantum error-correcting code given under the influence of any noisy quantum channel.
Note that $\tr P$ is the dimension of the quantum code and hence a constant in the linear program.

The independence of the auxiliary weight enumerator on the underlying quantum channel allows us to establish a single linear program for an 
entire family of quantum channels with respect to a fixed quantum code, and we illustrate this using the amplitude damping channel in the next section.

\section{Applications} \label{sec:applications}
\subsection{Amplitude damping errors} \label{subsec:AD}

AD errors model energy relaxation in quantum harmonic oscillator systems and photon loss in photonic systems.
By ensuring that each quantum harmonic oscillator couples identically to a unique bosonic bath, in the low temperature limit,
the effective noise model can be described by an AD channel.
When quantum information lies in a qubit, 
the corresponding AD channel 
$\cN_\gamma$ models energy loss in a two-level system, where $\gamma$ is the probability that an excited state relaxes to the ground state. $\cN_\gamma$ has two Kraus operators $K_0$ and $K_1$, where
\[
K_0=\begin{bmatrix}1&0\\0&\sqrt{1-\gamma}\end{bmatrix}, \quad K_1=\begin{bmatrix}0&\sqrt{\gamma}\\ 0&0\end{bmatrix}.
\]
When energy loss occurs independently and identically in an $n$-qubit system, 
the corresponding noisy channel can be modeled 
as $\cN_{n,\gamma}=\cN_{\gamma}^{\otimes n}$.
The set of all Kraus operators of $\cN_{n,\gamma}$ can be written as 
\begin{align}
\mathcal K =  \{K_{\bx}\triangleq K_{x_1}\otimes \cdots \ot K_{x_n}: \bx\in \zo^n\}.   
\end{align}
Since the Kraus operator $K_1$ models energy loss on one qubit, 
it is useful to know how many times the Kraus operator $K_1$ occurs in $K_\bx$.
Hence we define the following property of $K_\bx$.

\bd The  weight of $K_{\bx}$ for $x\in\{0,1\}^n$ is  $\wt{\bx}$.
\ed
The weight of $K_\bx$ counts the number of qubits where $K_\bx$ induces energy loss.
For example, $\wt{K_1\otimes K_0\ot K_1}=\wt{K_{101}}=\wt{101}= 2$, which corresponds to energy loss in two qubits.
Using this notion of weight, we partition the set of Kraus operators $\cK$ accordingly. Namely, by denoting 
\begin{align}
    \cK_i =  \{ E \in \cK : \wt{E} = i\},
\end{align}
we have $\cK = \cK_0 \cup \dots \cup \cK_n$.
In this terminology, a code corrects $t$ errors perfectly if all the errors in $\cK_i$ for $i \le t$ satisfy the Knill-Laflamme quantum error correction criterion \cite{KnL97}.

Specializing to the case of AD errors, our enumerators are vectors with coefficients
\begin{align}
    A_i    &= 
    \frac{1}{(\tr P)^2}
    \sum_{E \in \cK_i} \tr(EP)\tr(E^\dagger P),\ i=0,\dots,n, \\ 
    B_i     &= 
    \frac{1}{\tr P}
    \sum_{E \in \cK_i} \tr(EPE^\dagger P) , \ i=0,\dots,n.
\end{align}
The corresponding connection matrices are

\begin{align}
M_A\supewe &= 
    \sum_{i=0}^{n} \sum_{E \in \cK_i}
    \sum_{\sigma , \tau \in {\cG}_n}
    2^{-2n}
    \tr(E \sigma )
    \tr(E^\dagger \tau)
    |i\> \<\sigma|\<\tau|,\label{eq:MA-defi-AD}\\
M_B \supewe 
    &= 
    \sum_{i=0}^{n}   \sum_{E \in \cK_i}
    \sum_{\sigma , \tau \in {\cG}_n}
    2^{-2n}
    \tr(E \sigma E^\dagger \tau)
    |i\> \<\sigma|\<\tau|
    \label{eq:MB-defi-AD}.
\end{align}
From Section \ref{sec:Kraus-wtenum}, we know that $B_i\geq A_i \ge 0$ for all $i=0,\dots,n.$

Since the only Kraus operator in $\cK_0$ has a minimum singular value of $(1-\gamma)^{n/2}$ for $A_0$, we have the lower bound
\begin{align}
    A_0 \ge (1-\gamma)^n.
\end{align}
This is reminiscent of the scenario for SL weight enumerators, where we have $A\ssl_0=1.$

Furthermore, it is easy to see that every $B_i$ is at most 
$O(\gamma^i)$, Since the operator norm of Kraus operators from $\cK_i$ is $\gamma^{i/2}$, the operator norm of $EPE^\dagger$ for any $E \in \cK_i$ is at most $\gamma^i \tr(P)$. It follows from the H\"older inequality on the Hilbert-Schmidt inner product that 
\[|\tr(EPE^\dagger P )|= |\<EPE^\dagger ,P \>|\le \|EPE^\dagger \|\|P\|_1,\]
where $\|\cdot \|_1$ denotes the trace norm and $\|\cdot \|$ denotes the operator norm, which is the maximum singular value of a matrix.
Thus by counting the number of terms in $\cK_i$, we have
\begin{align}
    B_i/\gamma^{i} \le \binom n i. 
\end{align}
We can obtain another upperbound on $B_i$. Note that
\begin{align}
  B_i =& \frac{1}{\tr P}
    \sum_{E \in \cK_i } \tr(EPE^\dagger P) \notag\\
    \leq&  \frac{1}{\tr P} \sum_{E \in \cK_i} \tr(EPE^\dagger)\notag\\
     \leq&  \frac{1}{\tr P} \sum_{E \in \cK_i} \tr(EE^\dagger)\notag\\
     =&    \frac{1}{\tr P} {n \choose i} \tr(  \begin{bmatrix}1&0\\0&{1-\gamma}\end{bmatrix}^{\otimes n-i}\otimes \begin{bmatrix}{\gamma}&0\\ 0&0\end{bmatrix}^{\otimes i}
)\\
\le& 
\frac{1}{\tr P}
\binom n i 2^{n-i} \gamma^i. \label{eq:other-Bibound}
\end{align}

The quantum weight enumerators of $\cQ$ for AD channels are 
\begin{align}
|A \>  &=  A_0 |0\> + \dots + A_{n} |n\>, \notag \\
|B \>  &=  B_0 |0\> + \dots + B_{n} |n\> .
\end{align}
We have the following definition of AQEC criterion for AD channels, using the language of quantum weight enumerators.
\bd  \label{def:tc-AD-code}
An $((n,M))$ quantum code is called a $(t,c)$-AD code if
its quantum weight enumerators satisfy the constraints
\begin{align}
B_i-A_i\leq c \gamma^{t+1}, i=0,\dots,t, \label{eq:BleqA}
\end{align}
where $0 \le \gamma \le 1$.
\ed
In the language of Definition \ref{def:epsilon-code}, a $(t,c)$ AD code is
$(c \gamma^{t+1}, \dots, c \gamma^{t+1})$-AQEC with respect to the sets $\mathcal K_0,\dots,\mathcal K_{t}$.

\be
The four-qubit code in~\cite{LNCY97} has two logical codewords
\begin{align*}
\ket{0}_L=& \frac{1}{\sqrt{2}}\left(\ket{0000}+\ket{1111}\right),\\
\ket{1}_L=& \frac{1}{\sqrt{2}}\left(\ket{0011}+\ket{1100}\right).
\end{align*}
It has weight enumerators
\begin{align*}
A_0&=\gamma^4/64 - \gamma^3/4 + 5 \gamma^2/4 - 2 \gamma + 1;\\
A_1&=A_2=A_3=0;\\
A_4&=\gamma^4/64.
\end{align*}
\begin{align*}
B_0&=\gamma^4/16 - \gamma^3/4 + 5 \gamma^2/4 - 2 \gamma + 1;\\
B_2&=3 \gamma^4/8 - 3 \gamma^3/4 + 3 \gamma^2/4;\\
B_4&=\gamma^4/16;\\
B_1&=B_3=0.
\end{align*}
Therefore, this code cannot be a $(2,c)$-AD code for any $c>0$.
Conversely, this code is known to correct an arbitrary single AD error \cite{LNCY97}.

\ee

\be The weight enumerators of the nine-qubit Shor code are as follows.
Note that (\ref{eq:BgeqA}) holds here. In addition, by Definition~\ref{def:tc-AD-code},
this code cannot correct each  AD error of weight three. 
\begin{align*}
A_0=B_0 &=
1 - 9\gamma/2 + 153\gamma^2/16 \\
&\quad - 399\gamma^3/32 + 351\gamma^4/32 +  O(\gamma^5); \\
A_i = B_i &=0, \quad i=1,2,4,5,7,8;\\
A_3 = A_9 &=0;\\
B_3 & = 3\gamma^3/4 + O(\gamma^4);\\
B_9 & =  \gamma^9/32;\\
A_6 & =  3\gamma^6/16 + - 9\gamma^7/32 + 45\gamma^8/256 + O(\gamma^9);\\
B_6 & =  3\gamma^6/16  - 9\gamma^7/32  + 9\gamma^8/32   + O(\gamma^9).
\end{align*}
Since the  leading order of $B_3-A_3$  in $\gamma$ is cubic, the Shor code cannot be a $(3,c)$-AD code for any $c>0$.
Hence, this is consistent with the fact that the Shor code corrects two AD errors~\cite{gottesman-thesis}.
\ee

From the above discussion, the weight enumerators $\ket{A}$ and $\ket{B}$ of a $(t,c)$-AD code must satisfy (\ref{eq:BleqA}), (\ref{eq:MAewe}), and (\ref{eq:MBewe}). We formulate a linear program with a constant objective function, and find non-negative variables $A_0,\dots,A_n$, $B_0,\dots,B_n$ that belong to a particular feasible region. 
The feasibility problem of our linear program is then equivalent to the following. 
\begin{align}
    {\rm Find}\ \ A_0,\dots,&A_n,B_0,\dots, B_n, \psiew\notag\\
      {\rm subject\ to\ } 
(\tr P)^2|A  \>      &=M_A 
\psiew \notag\\
\tr P|B  \> &=M_B   \psiew\notag\\
(B_i    - A_i)/\gamma^{t+1}  &\le c  , \quad 0 \le i \le d-1\notag\\ 
B_i    /\gamma^i  &\le \binom n i  , \quad 0 \le i \le n\notag\\ 
B_i    /\gamma^i  &\le \binom n i 2^{n-i} / \tr P  , \quad 0 \le i \le n\notag\\ 
B_0+\dots +B_n      &\le 1 \notag\\ 
A_i      &\ge 0, \quad 0 \le i \le n  \notag\\
\Pi_{} \psiew &= \psiew. \label{eq:linear-constraints-AD}
\end{align}
Note that the constraint $B_i    /\gamma^i  \le \binom n i 2^{n-i} / \tr P$ only becomes non-trivial for large values of $\tr P$.

Since integer programs are hard to solve in general, 
our feasibility conditions are attractive because they have no integer constraints, in contrast to many other linear programming bounds for stabilizer codes \cite{CRSS98,LBW13,LA18,ALB20}. 
Hence, we have a linear program as opposed to an integer program.
However, one may wonder whether such a linear program is sufficiently constrained to be potentially infeasible.
We demonstrate numerically that our linear program can be infeasible, by analyzing the potential of using three qubits to correct a single AD error.
To do this, we have an additional observation that our linear program is parametrized by $\gamma$.
Since a $(t,c)$-AD code is defined for any value of $\gamma$ in the unit interval, we can concatenate the linear constraints using many different values of $\gamma$. 
Crucially, constraints for different values of $\gamma$ are related because $\psiew$ is independent of $\gamma$. We illustrate the linear dependence of all of our linear constraints in Fig.~\ref{fig:star}.

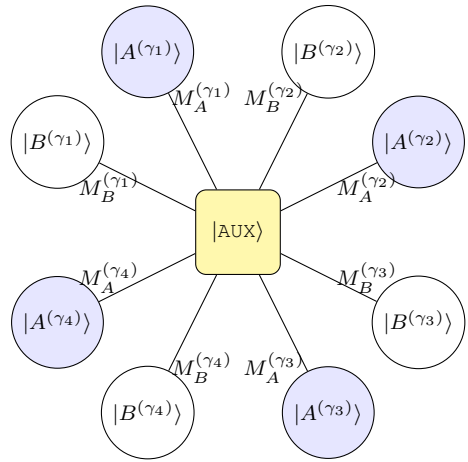
\begin{figure}
\begin{center}
		\begin{tikzpicture}[scale=1.2][thick]
		\fontsize{8pt}{1} 
		\tikzstyle{checknode} = [draw,fill=blue!10,shape= circle,minimum size=0.4em]
		\tikzstyle{variablenode} = [draw,fill=white, shape=circle,minimum size=0.4em]
		\tikzstyle{measurenode} = [draw,fill=yellow!40,shape= rectangle,rounded corners, minimum size=4.0em]
		\node[measurenode] (mn1) at (0,0) {$\ket{{\texttt{AUX}}}$} ;
		\node[checknode] (cn1) at (-1,2) { $\ket{A^{(\gamma_1)}}$} ; ;
		\node[variablenode] (vn4) at (-1,-2) { $\ket{B^{(\gamma_4)}}$} ;
		\node[variablenode] (vn2) at (1,2) {$ \ket{B^{(\gamma_2)}}$} ;
		\node[variablenode] (vn3) at (2,-1) {$\ket{B^{(\gamma_3)}}$} ;
			\node[checknode] (cn3) at (1,-2) { $\ket{A^{(\gamma_3)}}$} ; ;
		\node[checknode] (cn4) at (-2,-1) { $\ket{A^{(\gamma_4)}}$} ;
		\node[checknode] (cn2) at (2,1) {$ \ket{A^{(\gamma_2)}}$} ;
		\node[variablenode] (vn1) at (-2,1) {$\ket{B^{(\gamma_1)}}$} ;
		\node (label1) at (-0.4,1.5) {$M_A^{(\gamma_1)}$} ;
	    \node (label2) at (0.4,1.5)  {$M_B^{(\gamma_2)}$}; 
	    \node (label1) at (-0.4,-1.5) {$M_B^{(\gamma_4)}$} ;
	    \node (label2) at (0.4,-1.5)  {$M_A^{(\gamma_3)}$}; 
		%
		\draw (cn1) --  (mn1) (vn1) --node[midway,left]{$M_B^{(\gamma_1)}$} (mn1)  (vn2) --(mn1) (cn2) --node[midway,right]{$M_A^{(\gamma_2)}$} (mn1);
		\draw (cn3) --  (mn1) (vn3) --node[right]{$M_B^{(\gamma_3)}$} (mn1)  (vn4) -- (mn1) (cn4) --node[midway,left]{$M_A^{(\gamma_4)}$} (mn1);
		%
		\end{tikzpicture}
\end{center}
	\caption{
	   The relationship between various enumerators is depicted here. Every A-type or B-type enumerator for differing values of AD parameter $\gamma_i$ relates linearly to the same auxiliary enumerator.
	  }\label{fig:star}
\end{figure}

To determine if our concatenated linear program is feasible, we code up the linear constraints in the \texttt{MATLAB} solver \texttt{cvx}, and use the algorithm SDPT3.  
In the linear constraints of \eqref{eq:linear-constraints}, we write the monomials of $\gamma$ as denominators. This normalizes our constraints so that a numerical solver can be numerically stable even for small values of $\gamma$.
Also, when coding up the linear constraints of \eqref{eq:linear-constraints} in a solver, 
we do not explicitly construct the permutation matrix $\Pi$ because it is much too big.
Rather we specify its implied linear constraints directly into the optimizer environment for our linear program.

In our numerical study, we analyze the possibility of correcting a single AD error using three qubits.
We obtain mainly no-go results on the existence of a three-qubits code that corrects a single AD error.
For this, we consider four different values of $\gamma$ in the construction of our linear program.
More precisely, we numerically find the maximum $c$ for which the convex solver returns a result that says that the linear program is infeasible. 

\begin{theorem}
There is no three-qubit $(1,9.8 \times 10^4)$-AD code that has dimension two.
\end{theorem}
\begin{proof}
For $n=3$, $M=2$, $t=1$, we rule out $c=9.8 \times 10^4$ using
$\gamma =  0.1 ,   0.05,    0.01,    0.0001$.
\end{proof}
Numerically, this value of $c=9.8 \times 10^4$ is the largest we could find for the parameters $n=3$, $M=2$, $t=1$.
If we could rule out three-qubit $(1,c)$-AD codes for all positive numbers $c$, then we would able to rule out all three-qubit codes that correct a single AD error.
\subsection{Linear programming bounds for permutation-invariant quantum codes}\label{subsec:picodes}
Permutation-invariant quantum codes are quantum codes that are invariant under any permutation of their underlying particles. 
Such codes have been studied in the qubit \cite{Rus00,PoR04,ouyang2014permutation,ouyang2015permutation}, the qudit \cite{OUYANG201743}, and the bosonic \cite{ouyang2019permutation} settings. These quantum codes are interesting because of not only  their capability to correct non-trivial errors such as quantum deletions \cite{HagiwaraISIT2020,ouyang2021permutation} and insertions \cite{shibayama2021equivalence}, but also   their potential applications as quantum memories \cite{ouyang2019mems} and for robust quantum metrology \cite{ouyang2019robust}. One key attractive feature of permutation-invariant quantum codes is the ease in which they can be prepared in physical systems \cite{wu2019initializing,johnsson2020geometric} as compared to the usual stabilizer codes.

Here, we restrict our attention to permutation-invariant quantum codes on qubits, and use linear programming methods to establish upper bounds on the minimum distance of permutation-invariant quantum codes of designed distances $d$.
A permutation-invariant quantum code has distance $d$ if it satisfies the Knill-Laflamme quantum error correction criterion that for every pair of orthogonal logical codewords $|i_L\>$ and $|j_L\>$ and for every $E \in \mathcal G_n$ of weight at most $d-1$,
\begin{align}
    \<i_L | E |j_L\> = g_E \delta_{i,j}
\end{align}
 for some complex coefficient $g_E$.

In this section, for   $i=0,1,\dots, n$, we set
\begin{align}
    \Omega_i = \{E \in \cG_n : \wt{E} = i\}.
\end{align}
Now let $P$ be a projector onto a permutation-invariant code. 
We are interested in the distance of  a permutation-invariant quantum code, and hence we use the usual SL-enumerators
$\{A\ssl_{i}\}$ and $\{B\ssl_{i}\}$.
While the $A$-type and $B$-type SL enumerators are related by the quantum MacWilliams identity, they are also related to our auxiliary weight enumerators using the corresponding connection matrices 
\begin{align}
M\ssl_A\supewe &= 
    \sum_{i=0}^{w-1} \sum_{E \in \Omega_i}
    \sum_{\sigma , \tau \in {\cG}_n}
    2^{-2n}
    \tr(E \sigma )
    \tr(E^\dagger \tau)
    |i\> \<\sigma|\<\tau|,\label{eq:MA-defi-SL}\\
M\ssl_B \supewe 
    &= 
    \sum_{i=0}^{w-1}   \sum_{E \in \Omega_i}
    \sum_{\sigma , \tau \in {\cG}_n}
    2^{-2n}
    \tr(E \sigma E^\dagger \tau)
    |i\> \<\sigma|\<\tau|
    \label{eq:MB-defi-SL}.
\end{align}
Now we will proceed to explain how we can impose permutation-invariant constraints on the auxiliary weight enumerator.
Note that for any qubit permutation $\pi$, we must have 
\begin{align}
P \pi = P  = \pi P,
\end{align}
for the projector $P$ of any permutation-invariant quantum code.
Then for any Pauli $\sigma$ and qubit permutation $\pi$, we  see that
\begin{align}
\tr(\pi \sigma \pi^\dagger P) = \tr(\sigma \pi^\dagger P \pi ) = \tr(\sigma P).
\end{align}
For non-negative integers $x,y,z$ such that $x+y+z \le n$, let
\begin{align}
\sigma_{x,y,z,n} = X^{\otimes x} \otimes Y^{\otimes y} \otimes Z^{\otimes z}  \otimes I^{\otimes (n-x-y-z)} .
\end{align}
Now define the sets
\begin{align}
C_{x,y,z,n} = \{ \pi \sigma_{x,y,z,n} \pi^\dagger : \pi \in S_n\},
\end{align}
where the symmetric group $S_n$ denotes the set of $n!$ permutations on $n$ qubits.
We can see that the set of Pauli operators $\mathcal G_n$ can be partitioned into the sets $C_{x,y,z,n}$.

We define an auxiliary enumerator for the permutation-invariant code
\begin{align}
|\phi_\picode\> = \sum_{0 \le x+y+z\le n} 
\tr(\sigma_{x,y,z,n}P)
|\sigma_{x,y,z,n}\>,
\end{align}
and let
\begin{align}
|{\rm AUX}_\picode\> = |\phi_\picode\>\otimes |\phi_\picode\>
\end{align}
denote the compressed auxiliary weight enumerator.  
Note that since $P$ is a Hermitian operator, it can always be expressed as a linear combination of Pauli matrices with real coefficients. Hence both 
$|\phi_{\picode}\>$ and
$|{\rm AUX}_\picode\>$ are real vectors.

A simple consequence of Lemma~\ref{lem:connection-matrices} is the following result.
\begin{lemma} 
\label{lem:connection-matrices-picodes}
Let $P$ be a code projector onto a permutation-invariant code, and let 
\begin{align}
W_n = \sum_{\substack{ 
0 \le x + y + z \le n\\ } 
}  \sum_{\tau \in C_{x,y,z,n}}
|\tau\>\<\sigma_{x,y,z,n}|.
\end{align}
denote a matrix with $4^n$ rows and $\binom{n+3}{3}$ columns. 
Then the following matrix identities hold.
\begin{align}
M\ssl_A (W_n \otimes W_n)   |{\rm AUX}_\picode\> &=(\tr P )^2 |A\ssl \>, 
\label{eq:MAewe-SL}\\
M\ssl_B (W_n \otimes W_n)   |{\rm AUX}_\picode\>&= \tr P    |B\ssl \>, \label{eq:MBewe-SL}  
\end{align}
where 
$|A\ssl\> = \sum_{j=0}^n A\ssl_j|j\>$
and
$|B\ssl\> = \sum_{j=0}^n B\ssl_j|j\>$.
\end{lemma}
\begin{proof}
Given any code projector $P$, we can write 
$|\phi\> = \sum_{\sigma \in \mathcal G_n }\tr(\sigma  P )|\sigma\>$. 
Now note that 
\begin{align}
W_n|\phi_\picode\>
= \sum_{0\le x+y+z\le n }
\sum_{\tau \in C_{x,y,z,n}}
\tr(\sigma_{x,y,z,n}P)
|\tau\>.
\end{align}
Since the code is permutation-invariant, we have that 
\begin{align}\tr(\tau P ) = \tr(\sigma_{x,y,z,n} P)\end{align}
for every 
$\tau \in C_{x,y,z,n}$.
Using this identity, we find that for permutation-invariant codes, we have 
\begin{align}
W_n|\phi_\picode\>
= \sum_{0\le x+y+z\le n }
\sum_{\tau \in C_{x,y,z,n}}
\tr(\tau P)
|\tau\> = |\phi\>.
\end{align}
Hence it follows that 
\begin{align}
& (W_n \otimes W_n ) |{\rm AUX}_\picode\>
 \notag\\
 =&
 ( W_n \otimes W_n ) |\phi_\picode\>
 \otimes
 |\phi_\picode\>
 \notag\\
 =&
 ( W_n |\phi_\picode\> ) \otimes 
 ( W_n |\phi_\picode\> ) \notag\\
 =&
 |\phi\> \otimes |\phi\>
 \notag\\
 =&|{\rm AUX}\> .
\end{align}
Substituting this into Lemma \ref{lem:connection-matrices} proves the result.
\end{proof}

Here, the number of columns in $W_n$ corresponds to the number of combinations of non-negative integers $x,y,z$ such that the constraint $x+y+z \le n$ is satisfied. The key difference between Lemma \ref{lem:connection-matrices} and 
Lemma \ref{lem:connection-matrices-picodes} is that the auxiliary weight enumerator for a permutation-invariant code has  dimension that is exponentially smaller than the dimension of $\psiew$. Namely, the dimension of 
$|\phi_\picode\>$ is $\binom{n+3}{3}$, which implies that the dimension of 
$|{\rm AUX}_\picode\>$ is $\binom{n+3}{3}^2 = \mathcal O(n^6),$ 
which grows only polynomially in $n$.

In what follows, we show that for permutation-invariant quantum codes, we can further compress the sizes of both the auxiliary weight enumerators and the connection matrices. Namely, in place of using $ |{\rm AUX}_\picode\> $, we can use  $|\phi_\picode\>$, and instead of using the connection matrices $M\ssl_A$ and $M\ssl_B$, we can use the following {\em compressed connection matrices} 
\begin{align}
    {\hat{M}}_A 
= &
    \sum_{\substack{ 0\le i\le n \\   } }
    \sum_{\substack{ 
 x + y + z = i\\   
 }} 
\binom{n}{x,y,z}    |i\> \<\sigma_{x,y,z,n}|,    
\end{align}
and
\begin{align}
    {\hat{M}}_B
&=
    \sum_{\substack{ 
    0\le i\le n \\ 
    0 \le x + y + z \le n} }
    F(i,x,y,z,n) 
    |i\>\<\sigma_{x,y,z,n}|,
\end{align}
where
\begin{align}
    F(i,x,y,z,n) = 
    \sum_{\substack{ 
E \in \cG_i \\
\sigma \in C_{x,y,z,n}\\
 }}
    2^{-n}
    \beta(E,\sigma),
\end{align}
and given any two Paulis $\sigma$ and $\tau$ in $\cG_n$,  $\beta(\sigma,\tau)=1$, if $\sigma$ and $\tau$ commute, and  
$\beta(\sigma,\tau)=-1$, otherwise. Note that 
\begin{align}
    F(0,x,y,z,n) = \binom{n}{x,y,z}2^{-n},
\end{align}
where $\binom{n}{x,y,z} = n!/(x!y!z!(n-x-y-z)!)$
is a multinomial coefficient.
We can furthermore exploit the symmetry of the summation in the definition of $F$ to get
\begin{align}
    F(i,x,y,z,n) = 
    \sum_{\substack{ 
a+b+c = i \\
\sigma \in C_{x,y,z,n}\\
 }}
    2^{-n}
    \beta(\sigma_{a,b,c,n},\sigma) \binom{n}{a,b,c}.\label{eq:F-simple-form}
\end{align}
From \eqref{eq:F-simple-form}, we can see that the complexity of evaluating $F$ is $\mathcal O(n^3 2^n)$.

We now present the following lemma, which shows how the $A$-type and $B$-type SL weight enumerators relate to one another via the compressed connection matrices $\cM_A$ and $\cM_B$.
For this, we define the permutation-invariant weight enumerator
\begin{align}
|\pi\> = 
\sum_{0 \le x+y+z \le n } \tr( \sigma_{x,y,z,n} P )^2 |\sigma_{x,y,z,n}\> 
\end{align}
\begin{lemma}
\label{lem:compressed-connection-matrices-picode}
\begin{align}
{\hat{M}}_A  |\pi\> &=      \tr(P)^2 |A\ssl\> \\
{\hat{M}}_B |\pi\>
&=  \tr(P) |B\ssl\> .
\end{align}
\end{lemma}
\begin{proof}
By the definition of ${\hat{M}}_A$, it follows that
\begin{align}
    {\hat{M}}_A |\pi\>
    =
    \sum_{i=0}^n \sum_{x+y+z=i}
    \binom n{x,y,z}
    |i\> \tr(\sigma_{x,y,z,n}P)^2.
\end{align}
We begin by simplifying the matrices $M\ssl_A (W_n \otimes W_n) $ and 
$M\ssl_B (W_n \otimes W_n) $. 
Note that 
\begin{align}
&
M\ssl_A (W_n \otimes W_n)   \notag\\
= &
    \sum_{\substack{ 0\le i\le n \\  E \in \cG_i  } }
    \sum_{\substack{ 
0 \le x + y + z \le n\\
0 \le a + b + c \le n\\
\sigma \in C_{x,y,z,n}\\
\tau \in C_{a,b,c,n}\\
 }}
    2^{-2n}
    \tr(E \sigma )
    \tr(E^\dagger \tau)
    |i\> \<\sigma_{x,y,z,n}|\<\sigma_{a,b,c,n}|  \notag\\
= &
    \sum_{\substack{ 0\le i\le n \\  
    x'+y'+z'=i\\
    E \in C_{x',y',z',n}  } }
    \sum_{\substack{ 
0 \le x + y + z \le n\\
0 \le a + b + c \le n\\
\sigma \in C_{x,y,z,n}\\
\tau \in C_{a,b,c,n}\\
 }}
   \delta_{E, \sigma}  \delta_{E, \tau}
    |i\> \<\sigma_{x,y,z,n}|\<\sigma_{a,b,c,n}|
    \notag\\
= &
    \sum_{\substack{ 0\le i\le n \\  
    x'+y'+z'=i\\
    E \in C_{x',y',z',n}  } }
    |i\> \<\sigma_{x',y',z',n}|\<\sigma_{x',y',z',n}|\notag\\
= &
    \sum_{\substack{ 0\le i\le n \\  
    x'+y'+z'=i\\
      } }
      | C_{x',y',z',n}|
    |i\> \<\sigma_{x',y',z',n}|\<\sigma_{x',y',z',n}|  \notag\\ 
= &
    \sum_{\substack{ 0\le i\le n \\   } }
    \sum_{\substack{ 
 x + y + z = i\\   
 }} 
\binom{n}{x,y,z}    |i\> \<\sigma_{x,y,z,n}|\<\sigma_{x,y,z,n}|.\notag
\end{align}
Since both 
$ M\ssl_A |{\rm AUX}_\picode\>$ and 
${\hat{M}}_A |\pi\>$ are equal to 
\begin{align}
\sum_{\substack{0 \le i \le n\\x+y+z =i}}
\binom n {x,y,z}
|i\> \tr(\sigma_{x,y,z,n} P)^2,
\end{align}
the first result of this lemma follows from Lemma \ref{lem:connection-matrices-picodes}.
Similarly,
\begin{align}
&M\ssl_B (W_n \otimes W_n)   \notag\\
= &
    \sum_{\substack{ 0\le i\le n \\  E \in \cG_i  } }
    \sum_{\substack{ 
0 \le x + y + z \le n\\
0 \le a + b + c \le n\\
\sigma \in C_{x,y,z,n}\\
\tau \in C_{a,b,c,n}\\
 }}
    2^{-2n}
    \tr(E \sigma  
     E^\dagger \tau)
    |i\> \<\sigma_{x,y,z,n}|\<\sigma_{a,b,c,n}|  \notag\\
= &
    \sum_{\substack{ 0\le i\le n \\  E \in \cK_i  } }
    \sum_{\substack{ 
0 \le x + y + z \le n\\
0 \le a + b + c \le n\\
\sigma \in C_{x,y,z,n}\\
\tau \in C_{a,b,c,n}\\
 }}
    2^{-2n}
    \delta_{\sigma,\tau}
    \tr(E \sigma  
     E^\dagger \tau)
    |i\> \<\sigma_{x,y,z,n}|\<\sigma_{a,b,c,n}|
    \notag\\
= &
    \sum_{\substack{ 0\le i\le n \\  E \in \cG_i  } }
    \sum_{\substack{ 
0 \le x + y + z \le n\\
\sigma \in C_{x,y,z,n}\\
 }}
    2^{-n}
    \varphi(E,\sigma)
    |i\> \<\sigma_{x,y,z,n}|\<\sigma_{x,y,z,n}| \notag\\
= &
    \sum_{\substack{ 
    0\le i\le n \\ 
    0 \le x + y + z \le n} }
    F(i,x,y,z,n) 
    |i\> \<\sigma_{x,y,z,n}|\<\sigma_{x,y,z,n}|.\notag
\end{align}
Since both 
$ M\ssl_B |{\rm AUX}_\picode\>$ and 
${\hat{M}}_B |\pi\>$ are equal to 
\begin{align}
\sum_{\substack{0 \le i \le n\\x+y+z =i}}
F(i,x,y,z,n)
|i\> \tr(\sigma_{x,y,z,n} P)^2,
\end{align}
the second result of this Lemma follows from Lemma \ref{lem:connection-matrices-picodes}.
\end{proof}
Using Lemma \ref{lem:compressed-connection-matrices-picode}, we use the effective auxiliary enumerator $|\phi_\picode\>$, which is of size $\mathcal O(n^3)$.

We now continue to introduce more constraints.
Note that for an $M$-dimensional  permutation-invariant code, its projector $P$  
admits the spectral decomposition \begin{align}
    P = \sum_{j=1}^M |L_j\>\<L_j|,
\end{align}
where $|L_j\>$ correspond to the logical codewords of the permutation-invariant quantum code.
We can expand every logical codeword $|L_j\>$ in the Dicke basis to get
\begin{align}
    |L_j\> 
    = \sum_{w=0}^n 
    a_{j,w} |D^n_w\>,
\end{align}
where $a_{j,w}$ are in general complex coefficients. 
We will restrict ourselves to permutation-invariant quantum codes where $a_{j,k}$ are non-negative.
This is a mild constraints because for every qubit $((n,M,d))$ permutation-invariant quantum code constructed so far \cite{ouyang2014permutation,OUYANG201743}, there exists an equivalent permutation-invariant quantum code that has these properties.

Now let us consider Paulis of the form $\sigma_{x,0,0,n}=X^{\otimes x} \otimes I^{\otimes n-x}$,
and analyze the properties of 
$\tr(\sigma_{x,0,0,n}P)$ when $x$ is even.
Note that
\begin{align}
    \tr(\sigma_{x,0,0,n}P)
    &=\sum_{j=0}^M
    \<L_j| \sigma_{x,0,0,n} | L_j\>\notag\\
    &\ge
    \sum_{j=0}^M
    \sum_{w=0}^n
    a_{j,w}^2
    \<D^n_w| \sigma_{x,0,0,n} | D^n_w\>\notag\\
    &\ge
    \sum_{j=0}^M
    \sum_{w=x/2}^{n-x/2}
    a_{j,w}^2
    \<D^n_w| \sigma_{x,0,0,n} | D^n_w\>\notag\\
    &=
    \sum_{j=0}^M
    \sum_{w=x/2}^{n-x/2}
    a_{j,w}^2
    \frac{\binom x {x/2} \binom {n-x}{w-x/2} }{\binom n w},
\end{align}
where 
the first and second inequalities arise because 
$  \<D^n_w| \sigma_{x,0,0,n} | D^n_{w'}\>$ and $a_{j,w}$ are all non-negative for all $w,w'=0,\dots,n$,
and in the last equality, we used a special case of \cite[Lemma 6]{ouyang2019robust}.

Using \cite[Lemma 2]{ouyang2014permutation},
we find that
\begin{align}
    \frac{\binom x {x/2} \binom {n-x}{w-x/2} }{\binom n w}
    &=
    \frac{\binom w {x/2}
    \binom {n-w}{x/2} }
    {\binom n x}.
\end{align}

Now denote 
\begin{align}
\beta_{n,x} 
    &= 
    \min_{x/2 \le w \le n-x/2}
    \binom w {x/2}
    \binom {n-w}{x/2} 
    /\binom n x.
\end{align}
Then 
\begin{align}
    &\tr(\sigma_{x,0,0,n}P)\notag\\
    \ge &
    \beta_{n,x}
    \sum_{j=0}^M
    \sum_{w=x/2}^{n-x/2}
    a_{j,w}^2\notag\\
    =&
    \beta_{n,x}
    \left(M - 
    \sum_{j=0}^M 
    \left(
       \sum_{w=0}^{x/2-1} a_{j,w}^2 +
       \sum_{w=n-x/2+1}^{n} a_{j,w}^2
    \right)
    \right).
\end{align}
The last equality arises from the normalization condition of the logical operators $|L_j\>$, which implies that $\sum_{w=0}^na_{j,w}^2 = 1$.

Using the fact that $a_{j,w}^2 \le 1$, we find that 
\begin{align}
    \tr(\sigma_{x,0,0,n}P)
    &\ge
    \beta_{n,x}
    \left(M - x   \right). \label{enum-lower-bound}
\end{align}
Note that \eqref{enum-lower-bound} is a strict inequality when $M=2, d>1$ and $(x,y,z)=(2,0,0)$. 
To see this, note that 
when \eqref{enum-lower-bound} holds with equality, this means that $\tr(\sigma_{2,0,0}P) = 0$. 
But the Dicke inner products $\<D^n_w| \sigma_{2,0,0}|D^n_w\>$ being positive for all $w=1,\dots, n-1$, implies that $P$ can only be supported on $|D^n_0\>=|0\>^{\otimes n}$ and $|D^n_n\>^{\otimes n}$. Since $M=2$, $P$ must be equal to $(|0\>\<0|)^{\otimes n} + (|1\>\<1|)^{\otimes n}$, and this is just the projector of the repetition code. The repetition code has a distance equal to 1, and this contradicts the premise that $d>1$.
Hence if $M=2$ and $(x,y,z)=(2,0,0)$, the inequality \eqref{enum-lower-bound} must be strictly positive.

\begin{theorem}
There is no [[5,1,3]] permutation-invariant code that has logical codewords with nonnegative  coefficients $a_{j,w}$ in the Dicke basis.
\end{theorem}
\begin{proof}
This is because we know that there is a unique solution to the SL enumerators for the [[5,1,3]] code~\cite{SL97}, which corresponds to 
$A\ssl_0=1, 
A\ssl_1=0,
A\ssl_2=0,
A\ssl_3=0,
A\ssl_4=15,
A\ssl_5=0.$
But \eqref{enum-lower-bound} holding strictly implies that $A\ssl_2>0$, and hence the linear program for [[5,1,3]] permutation invariant codes must be infeasible, and hence a [[5,1,3]] permutation-invariant code.
\end{proof}

For a permutation-invariant code with nonnegative $a_{j,w}$ to have minimum distance of $d$, the SL enumerators need to satisfy not only the usual MacWilliams identity, but also the additional constraints related to the auxiliary weight enumerator for permutation-invariant codes. 
Let 
\begin{align}
 T_2 =    
 \sum_{\substack{
 x,y \ {\rm even}\\
 0 \le x \le n\\
 0 \le y \le n\\
 x+y \le n\\
 }
 }
 |\sigma_{x,y,0}\>\<\sigma_{x,y,0}|,
\end{align}
and
\begin{align}
 \tau_2 =    
 \sum_{\substack{
 x,y \ {\rm even}\\
 0 \le x \le n\\
 0 \le y \le n\\
 x+y \le n\\
 }
 }
 \beta_{n,x+y}^2(M-x-y)^2
 |\sigma_{x,y,0}\>.
\end{align}
We formulate the linear program that maximizes $A\ssl_2$ subject to the following constraints:
\begin{align}
    {\rm Find}\ \ A\ssl_0,\dots,&A\ssl_n,
    B\ssl_0,\dots, B\ssl_n  
    \ge 0   \notag\\
    |\pi\>  & \in  \mathbb R^{(n+1)(n+2)(n+3)/6}\notag\\
      {\rm subject\ to\ } 
(\tr P)^2|A\ssl  \>      &={\hat{M}}_A |\pi\> \notag\\
\tr P|B\ssl  \> &=  {\hat{M}}_B |\pi\>\notag\\
B\ssl_i    - A\ssl_i  & = 0  , \quad 0 \le i \le d-1\notag\\   
B\ssl_i &\ge A\ssl_i,   \quad d \le i \le n  \notag\\
 |\pi\> &\ge 0\notag\\
T_2|\pi\> &\ge \tau_2   \notag\\
   |\pi\> &\le \tr(P)^2   \notag\\
B\ssl_i&= \frac{\tr(P)}{2^n}\sum_{j=0}^n A\ssl_{j} K_{i}(j;n), \ i=0,\dots,n. \label{eq:LP}
\end{align}
The last equality constraints in program~(\ref{eq:LP}) are from the MacWilliams identity~\cite{SL97} and \begin{align}
K_i(x;n)= \sum_{j=0}^i (-1)^j 3^{i-j}{x\choose j}{n-x \choose i-j}\label{kraw}
\end{align}
is the $i$-th quaternary Krawtchouk polynomial.
Also let 
\begin{align}
    M = 
    \frac{\tr(P)}{2^n}
    \sum_{i=0}^n
    \sum_{j=0}^n
    K_i(j;n)
    |i\>\<j| .
\end{align}
If the linear program is infeasible, then we know for sure that there do not exist $((n,M,d))$ permutation-invariant quantum codes that have logical codewords with nonnegative $a_{j,w}$.
Furthermore, for an $((n,2,d \ge 2))$ permutation-invariant quantum code to exist, 
we know that $A\ssl_2$ must be strictly positive, 
because if $A\ssl_2 = 0$, we must have $\sigma_{2,0,0,n}=0$ which implies that the permutation-invariant code must be a repetition code with $d=1$, which contradicts the assumption that $d \ge 2$.

To aid the running our linear programm numerically using the \texttt{linprog} function of \texttt{MATLAB} so that the linear program can be evaluated using the simplex algorithm, we write the constraints of our linear program in standard form.
Let $I_d = \sum_{j=0}^{d-1}
|j\>\<j|$ and $\bar I_d = 
\sum_{j=d}^n|j\>\<j|$.
Then our equality constraints are
\begin{align}
\begin{pmatrix}
\<0| & 0 & 0 \\
0 & 0 & \<\sigma_{0,0,0}| \\
M & -I & 0 \\
I_d& -I_d& 0  \\
-\tr(P)^2I & 0 & {\hat{M}}_A \\
0& -\tr(P)I & {\hat{M}}_B \\
\end{pmatrix}
    \begin{pmatrix}
    A\ssl \\
    B\ssl \\
    |\pi\> \\
    \end{pmatrix}
    = 
    \begin{pmatrix}
    1\\ \tr(P)^2\\0\\0\\0 \\0\\
    \end{pmatrix}
\end{align}
and the inequality constraints are
\begin{align}
    \begin{pmatrix}
    \bar I_d &
    -\bar I_d &
    0\\
    0 & 0 & I\\
    0& 0& -T_2
    \\
    \end{pmatrix}
    \begin{pmatrix}
    A\ssl \\
    B\ssl \\
    |\pi\> \\
    \end{pmatrix}
    \le
    \begin{pmatrix}
    0 \\ \tr(P)^2 \\ 
    -\tau_2
    \end{pmatrix}.
\end{align}

Finally there is a trivial upper bound that
\begin{align}M \le n+1,
\end{align}
since the dimension of a permutation-invariant quantum code cannot exceed the dimension of the symmetric subspace.
Using our linear programming bounds, together with the trivial bound,  we tabulate upper bounds on permutation-invariant codes that have logical codewords with nonnegative $a_{j,w}$ in Table \ref{tab:pibounds}.
From this table, we can see that to have $d=2$ we need at least $n\ge 4$; for $d=3$, we need $n\ge 3$; for $d=4$, we need $n\ge 8$; and for $d=5$, we need $n \ge 11$.
In comparison, we 
give upper bounds on $M$ for general permutation-invariant codes using only the MacWilliams identities as the constraints, and the trivial $M\le n+1$ bound in Table \ref{tab:general-pibounds}.
We also give known lower bounds for $M$ using  code constructions for permutation-invariant quantum codes that have logical codewords with nonnegative $a_{j,w}$ in Table \ref{tab:lower-pibounds}.

\begin{table}[]
\begin{tabular}{l|llllllllll}
\hline
    &   &   &   &   & $n$ &  &   &   &   &   \\ 
    & 3   &   4 &   5 &   6 &   7 &   8 &   9 &   10 &  11 & 12 \\
    \hline
$d=2$ & 1   &  3  & 5   & 7   & 8   &  9  & 10  & 11  & 12  & 13 \\
$d=3$ &     &     & 1   & 2   & 4   &  6  & 10  & 11   & 12   & 13  \\
$d=4$ &     &     &     &     & 1   &  2  & 4   & 6    & 9   & 13   \\
$d=5$ &     &     &     &     &     &  1  & 1   & 1    & 2    & 4    \\
$d=6$ &     &     &     &     &     &     &     &      & 1    & 1   
\end{tabular}
\caption{Table of upper bounds of $M$ for $((n,M,d))$ permutation-invariant quantum codes that have logical codewords with non-negative $a_{j,w}$. We use our linear programming bounds together with the trivial constraint $M \le n+1$.}
\label{tab:pibounds}
\end{table}

\begin{table}[]
\begin{tabular}{l|llllllllll}
\hline
    &   &   &   &   & $n$ &  &   &   &   &   \\
    & 3   &   4 &   5 &   6 &   7 &   8 &   9 &   10 &  11 & 12 \\
    \hline
$d=2$ & 2   & 4   & 6   & 7  &  8   & 9   & 10  &   11 &  12  & 13 \\
$d=3$ &     &     & 2   & 2   & 4   & 9   & 10  & 11   & 12   & 13  \\
$d=4$ &     &     &     &     & 1   & 2   & 4   & 9    & 12   & 13   \\
$d=5$ &     &     &     &     &     &     & 1   & 2    & 3    & 5    \\
$d=6$ &     &     &     &     &     &     &     &      & 1    & 1   
\end{tabular}
\caption{Table of upper bounds for $M$ for $((n,M,d))$ arbitrary permutation-invariant quantum codes using the quantum MacWilliam identities together with $M\le n+1$. }
\label{tab:general-pibounds}
\end{table}

\begin{table}[]
\begin{tabular}{l|llllllllll}
\hline
    &   &   &   &   & $n$ &  &   &   &   &   \\
    & 3   &   4 &   5 &   6 &   7 &   8 &   9 &   10 &  11 & 12 \\
    \hline
$d=2$ &     & 2$^g$   & 2$^g$   & 2$^g$  &  3$^*$   & 3$^*$   & 3$^*$   &  3$^*$  &  3$^*$     & 4$^*$   \\
$d=3$ &     &     &     &     & 2$^p$   &     & 2$^{rg}$  & 2$^g$   &  2$^g$   &  2$^g$  \\
$d=4$ &     &     &     &     &     &     &     &      &     &      \\
\end{tabular}
\caption{Table of lower bounds for $M$ for $((n,M,d))$ permutation-invariant quantum codes with non-negative $a_{j,w}$. The superscript $^g$ referes to the gnu codes introduced in \cite{ouyang2014permutation}, $^r$ refers to the Ruskai code \cite{Rus00}, $^{p}$ refers to the Pollatsek-Ruskai 7-qubit code, and $^*$ refers to the codes in \cite{OUYANG201743}.
}
\label{tab:lower-pibounds}
\end{table}

\section{Discussions} \label{sec:discussions}
In this paper, we showed that quantum weight enumerators can be generalized to the setting of AQEC. Key to our analysis is our introduction of auxiliary weight enumerators, which allows us to establish an indirect linear relationship between the generalized quantum weight enumerators.

As it stands, the auxiliary weight enumerator is a vector of size $4^{2n}$ in the number of qubits $n$.
We have shown how exploiting the symmetry of permutation-invariant quantum codes can greatly reduce the dimensionality of the auxiliary weight enumerator to have a size that is polynomial in $n$.
Specializing our framework to a broad family of permutation-invariant quantum codes, we use linear programming to obtain non-trivial upper bounds on the maximum number of logical codewords $M$ of such permutation-invariant codes for given length $n$ and distance $d$.

\section{Acknowledgements}\label{eq:acknow}
YO acknowledges support from the EPSRC (Grant No. EP/M024261/1) and the QCDA project (Grant No. EP/R043825/1)) which has received funding from the QuantERA ERANET Cofund in Quantum Technologies implemented within the European Union’s Horizon 2020 Programme.
YO is also supported in part by NUS startup
grants (R-263-000-E32-133 and R-263-000-E32-731), and the
National Research Foundation, Prime Minister ’s Office, Singapore
and the Ministry of Education, Singapore under the
Research Centres of Excellence programme.
CYL was  supported  by the Ministry of Science and Technology (MOST) in Taiwan, under Grant  MOST109-2636-E-009-004 and  Grant MOST110-2628-E-A49 -007.

\bibliography{_paper2}{}
\bibliographystyle{ieeetr}

\end{document}